\documentclass[journal]{IEEEtran}
\newcommand{\bl}{\textcolor{black}}

\usepackage{graphicx,epsfig,amssymb,amsmath,url,latexsym,stfloats,array,bm}
\usepackage[tight,footnotesize]{subfigure}
\usepackage{makecell}
\usepackage{mathrsfs}
\usepackage{bbm}
\usepackage{amsthm}
\usepackage{longtable}
\usepackage{amsfonts}
\usepackage{caption}
\usepackage{graphicx}
\usepackage{multirow}
\usepackage{longtable}
\usepackage{float}
\usepackage{afterpage}
\usepackage{cite}
\usepackage{subeqnarray}
\usepackage{epstopdf}
\usepackage{empheq}
\usepackage{latexsym}
\usepackage{CJK}
\usepackage{color, soul}
\usepackage{framed}
\usepackage{lipsum}
\usepackage{color}
\usepackage{stfloats}
\usepackage{setspace}
\definecolor{shadecolor}{rgb}{1,0,0}

\newtheorem{lemma}{\textbf{Lemma}}

\usepackage{algorithmic}
\usepackage[linesnumbered,ruled]{algorithm2e}
\usepackage{cuted}

\begin{document}
\title{Digital Twin-Based 3D Map Management for Edge-Assisted Mobile Augmented Reality}

{\setstretch{1.0}
	\author{
	\IEEEauthorblockN{Conghao~Zhou\IEEEauthorrefmark{1},~Jie~Gao\IEEEauthorrefmark{2},~Mushu~Li\IEEEauthorrefmark{3},~Nan~Cheng\IEEEauthorrefmark{4},~Xuemin (Sherman)~Shen\IEEEauthorrefmark{1},~and~Weihua Zhuang\IEEEauthorrefmark{1}}
	    \IEEEauthorblockA{\IEEEauthorrefmark{1}Department~of~Electrical~and~Computer~Engineering,~University~of~Waterloo,~Canada
	    \\\IEEEauthorrefmark{2}School of Information Technology, Carleton University,~Canada
	    \\\IEEEauthorrefmark{3}Department of Electrical, Computer, and Biomedical Engineering, Toronto Metropolitan University,~Canada
	    \\\IEEEauthorrefmark{4}School~of~Telecommunication Engineering,~Xidian University,~Xi'an,~China
	    \\\{c89zhou, sshen, wzhuang\}@uwaterloo.ca, jie.gao6@carleton.ca, mushu1.li@ryerson.ca, dr.nan.cheng@ieee.org}
			}
}

\maketitle

\begin{abstract} 

In this paper, we design a 3D map management scheme for edge-assisted mobile augmented reality (MAR) to support the pose estimation of individual MAR device, which uploads camera frames to an edge server. Our objective is to minimize the pose estimation uncertainty of the MAR device by periodically selecting a proper set of camera frames for uploading to update the 3D map. To address the challenges of the dynamic uplink data rate and the time-varying pose of the MAR device, we propose a digital twin (DT)-based approach to 3D map management. First, a DT is created for the MAR device, which emulates 3D map management based on predicting subsequent camera frames. Second, a model-based reinforcement learning (MBRL) algorithm is developed, utilizing the data collected from both the actual and the emulated data to manage the 3D map. With extensive emulated data provided by the DT, the MBRL algorithm can quickly provide an adaptive map management policy in a highly dynamic environment. Simulation results demonstrate that the proposed DT-based 3D map management outperforms benchmark schemes by achieving lower pose estimation uncertainty and higher data efficiency in dynamic environments.

\end{abstract}

\section{Introduction}

\bl{With future sixth-generation (6G) networks,} immersive communications are anticipated to revolutionize the current communication paradigm by providing users with highly realistic and interactive experiences~\cite{shen2021holistic}. As a representative type of immersive communications, mobile augmented reality (MAR) aims to seamlessly integrate virtual objects into physical environments, enabling users to interact with lifelike virtual objects superimposed on their immediate surroundings via MAR devices. In the 6G era, MAR is expected to play an important role in various applications, such as education and tourism, due to its potential in boosting user engagement~\cite{ahn2020novel}.

In MAR applications, attaining to high immersion necessitates the accurate geometric alignment of virtual objects with each MAR device's physical environment. To achieve this, a 3D map consisting of the 3D positions of feature points (e.g., distinguishable landmarks) in the physical environment is commonly utilized, in which each feature point corresponds to a 3D \emph{map point}~\cite{ran2019sharear}. By matching the 3D map points detected in the current camera frame captured by an MAR device with those contained in the 3D map, the pose of the MAR device, as well as the spatial layout and geometric information about its immediate surroundings, can be estimated. This procedure can establish the accurate perspective and spatial relationships required for anchoring virtual objects to the 3D map points detected in each camera frame. As a result, including a proper set of 3D map points in the 3D map is vital for accurate pose estimation~\cite{campos2021orb}. Nevertheless, managing (i.e., building and updating) a 3D map in an MAR application can be both computation- and data-intensive due to the need for processing and storing a large amount of 3D map points in each camera frame, which poses a significant challenge to MAR devices with limited resources~\cite{cheng2019space, ma2023nomore}.

Mobile edge computing for MAR applications, which can leverage the resources of edge servers and the ubiquitous connectivity provided by communication networks, can assist MAR devices in managing their 3D maps~\cite{ahn2020novel}. In edge-assisted MAR, camera frames captured by an MAR device can be uploaded to a nearby edge server, and the 3D map of the MAR device can be built and updated by the edge sever based on the uploaded camera frames. While MEC can alleviate the computing and data storage burdens on MAR devices, effectively managing a 3D map for individual MAR device still faces two key challenges. First, each MAR device may frequently change its pose, i.e., its position and orientation, leading to a time-varying set of 3D map points in its camera frame sequence. Moreover, the temporal variations in the poses of two MAR devices may be different even if they are in the same environment~\cite{chen2018marvel}. Thus, for each MAR device, building a customized 3D map and timely updating the 3D map to cope with its unique time-varying pose are challenging. Second, given the limited uplink communication resource of an MAR device, selecting the optimal subset of camera frames for uploading to update the 3D map without compromising pose estimation accuracy is another challenge, which is aggravated by the MAR device's dynamic uplink data rate~\cite{apicharttrisorn2020characterization}. 

Several studies have made efforts to address the aforementioned challenges. Han~\emph{et~al.} concentrated on the proactive caching of virtual objects by modeling the 2D trajectory of an MAR device as a discrete Markov chain~\cite{han2022intelli}. In the context of edge-assisted MAR, Dhakal~\emph{et~al.} investigated managing a 3D map for multiple users under the assumption that MAR devices can upload all camera frames without being limited by resource constraints~\cite{ran2019sharear}. Chen~\emph{et~al.} derived the performance metric, i.e.,~\emph{uncertainty}, to quantify the impact of a set of 3D map points on pose estimation accuracy~\cite{chen2023adaptslam}. However, existing 3D maps in MAR are myopically updated based on the current uplink data rate and pose of each MAR device, which may not adapt to highly dynamic environments.

In this paper, we design a digital twin (DT)-based 3D map management scheme for edge-assisted MAR to adapt to the time-varying pose and the dynamic uplink data rate of individual MAR device subject to resource constraints. An edge server builds a 3D map to support the pose estimation of the MAR device which can also upload camera frames to the edge server for updating the 3D map. Our objective is to minimize the long-term pose estimation uncertainty of the MAR device by periodically selecting a proper subset of cameras frames to upload for 3D map update. The main contributions of this paper are as follows: 

	\begin{itemize}
		\item We design a digital twin (DT)-based 3D map management scheme. \bl{A DT created for an MAR device can assist in predicting the 3D map points in future camera frames, thereby supporting the emulation of 3D map management and providing emulated data to facilitate personalized 3D map management.}

		\item We develop an adaptive and data-efficient 3D map management algorithm featuring model-based reinforcement learning (MBRL). By leveraging the combination of the real data from actual 3D map management and the emulated data from the DT, the algorithm can quickly provide an adaptive 3D map management policy in highly dynamic environments.

	\end{itemize}


\section{System Model and Problem Formulation}

In this section, we introduce the system model and formulate the 3D map management problem.

\subsection{Considered Scenario}

    \begin{figure}[t]
        \centering
        \includegraphics[width=0.42\textwidth]{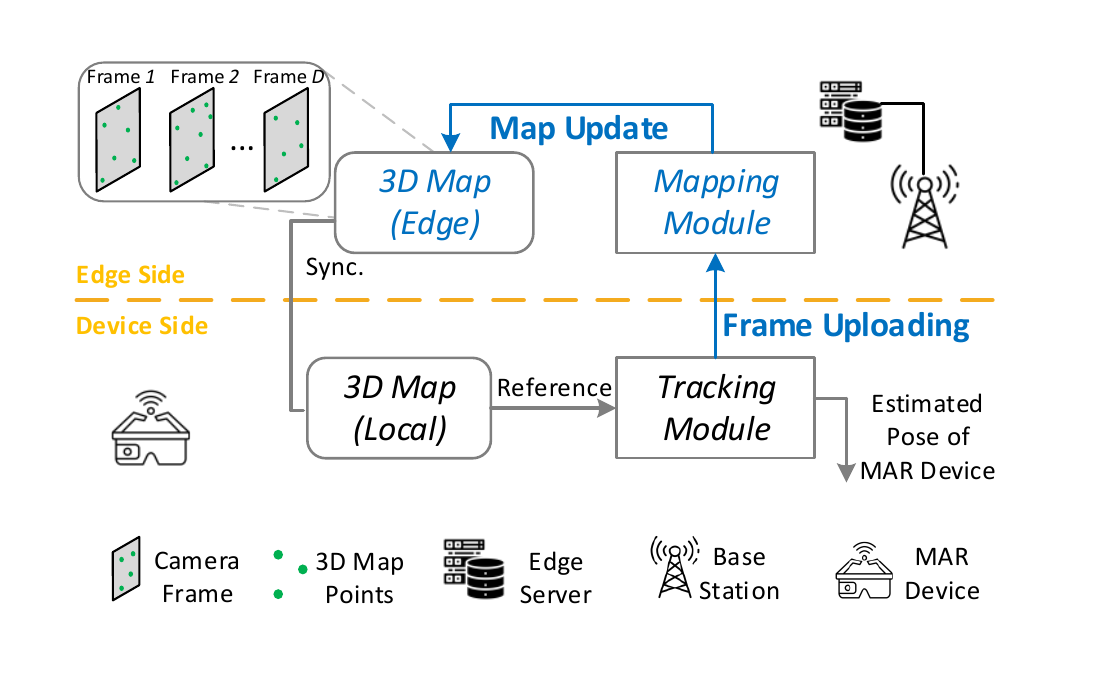}
        \caption{The considered scenario of edge-assisted MAR.}\label{system}
        \vspace{-0.5cm}
    \end{figure}

Consider the scenario of one MAR device running an indoor MAR application. The surrounding environment of the MAR device does not change over time~\cite{cozzolino2022nimbus}. The MAR device periodically captures camera frames and needs to estimate its real-time pose for aligning virtual objects with its surrounding environment. As shown in Fig.~\ref{system}, we adopt the architecture of edge-assisted MAR, which includes a tracking module and a mapping module~\cite{campos2021orb}, to support real-time pose estimation. The MAR device, containing the tracking module, estimates its pose for each frame based on a 3D map and can upload frames to an edge server deployed at a base station via a wireless link. The edge server, containing the mapping module, assists the MAR device in building and updating the 3D map based on the uploaded frames. The workflow includes four steps: 
	\begin{enumerate}
		\item \emph{Real-time pose estimation}: The MAR device estimates its pose for each frame by matching the 3D map points (i.e., 3D feature points in the surrounding environment) detected in this frame with those contained in the local 3D map (i.e.,~``3D Map (Local)'' in Fig.~\ref{system});

		\item \emph{Frame uploading}: The MAR device can select only a subset of all captured frames for uploading due to resource constraints;

		\item \emph{3D map update}: The edge server updates the 3D map (i.e.,~``3D Map (Edge)'' in Fig.~\ref{system}) by processing the uploaded frames;

		\item \emph{Synchronization}: The updated 3D map is sent back to the MAR device for subsequent pose estimation.
	\end{enumerate}
The real-time pose estimation (Step~1) operates on each frame and takes as short as several milliseconds to complete, while 3D map management (Steps~2-4) generally operates on a larger time scale (e.g., over seconds)~\cite{campos2021orb}. In this paper, we focus on~\emph{frame uploading} and \emph{3D map update} (corresponding to the blue arrows in Fig.~\ref{system}), detailed in Subsections~\ref{sec22} and~\ref{sec23}, respectively. Consider that the edge server manages the 3D map periodically, and refer to each period as a time slot. Denote the index set of all time slots and the set of all frames captured across all time slots by~$\mathcal{T}$ and~$\mathcal{F}$, respectively.

\subsection{Frame Uploading}\label{sec22}

Properly selecting a subset of frames from all captured frames during each time slot for uploading is crucial for the performance of 3D map management. A mmWave band featuring a high data rate is utilized to upload frames. Meanwhile, a sub-6\,GHz band with a relatively low data rate is used as a backup to ensure uninterrupted uplink communication in the case of mmWave band outage. We assume a constant uplink data rate within any given time slot and model the dynamic change of the uplink data rate across different time slots due to channel switching as a two-state Markov chain. Denote the probabilities of the high-rate state transiting from and to the low-rate state by $p_\text{l}$ and $p_\text{h}$, respectively. Without loss of generality, the amount of data (in bits) for uploading each frame, denoted by~$d$, is assumed to be identical. Denote the set of all frames captured during time slot~$t \in \mathcal{T}$ and the subset of frames selected for uploading by $\mathcal{F}_{t} \subseteq \mathcal{F}$ and~$\mathcal{U}_{t} \subseteq \mathcal{F}_{t}$, respectively. Determining the set~$\mathcal{U}_{t}$ for uploading should satisfy the following constraint:
    \begin{equation}\label{eq1}
        |\mathcal{U}_{t}| d \leq R_{t} \tau, \,\, \forall t \in \mathcal{T},
    \end{equation}
where $R_{t}$ denotes the uplink data rate (in bits per second) during time slot~$t$ and equals either~$R_\text{h}$ or~$R_\text{l}$, $\tau$ denotes the duration of each time slot, and $|\cdot|$ represents the cardinality of a set.

\subsection{3D Map Update}\label{sec23}

In this subsection, we first develop the model of 3D map in MAR and then introduce the 3D map update.

As shown in Fig.~\ref{system}, the edge server builds a 3D map by storing multiple frames, each of which corresponds to a set of 3D map points detected in this frame~\cite{campos2021orb,ran2019sharear}. As a result, the 3D map includes the union of multiple sets of 3D map points, in which each 3D map point corresponds to one or more frames. We model the 3D map as an undirected graph, denoted by~$\mathcal{G}^\text{e}_{t} = (\mathcal{D}^\text{e}_{t}, \mathcal{E}^\text{e}_{t})$ where $\mathcal{D}^\text{e}_{t}$ and $\mathcal{E}^\text{e}_{t}$ are the sets of nodes and edges, respectively. The set of nodes~$\mathcal{D}^\text{e}_{t} \in \mathcal{F}$ corresponds to the set of frames stored in the 3D map at the beginning of time slot~$t$. The set of edges~$\mathcal{E}^\text{e}_{t}$ corresponds to the set of co-visibility relationship between any two frames in~$\mathcal{D}^\text{e}_{t}$, i.e., the intersection of their respective sets of 3D map points~\cite{chen2023adaptslam}. Denote the edge between frames~$f, f' \in \mathcal{D}^\text{e}_{t}$ by~$e = (f, f') \in \mathcal{E}^\text{e}_{t}$. The weight of edge~$e$, denoted by~$w_{e}$, represents the number of 3D map points found in both frame~$f$ and frame~$f'$, given by:
	\begin{equation}\label{}
        w_{e} = |\mathcal{M}_{f} \cap \mathcal{M}_{f^{'}}|, \,\, \forall e \in \mathcal{E}^\text{e}_{t},
    \end{equation}
where $\mathcal{M}_{f}$ denotes the set of 3D map points corresponding to frame~$f \in \mathcal{F}$.

The edge server first updates the set of frames stored in the 3D map upon receiving the uploaded frames.\footnote{The phrases ``update a 3D map'' and ``update the set of frames contained in a 3D map'' are used interchangeable in this paper.} Since the computing and storage resources of the edge server are constrained for an individual MAR device, the following constraint should be satisfied for 3D map update:
    \begin{equation}\label{eq2}
        |\mathcal{D}^\text{e}_{t}| \leq D, \,\, \forall t \in \mathcal{T},
    \end{equation}
where $D$ denotes the maximum number of frames that can be processed and stored by the edge server for the MAR device. Given the set of newly uploaded frames,~i.e.,~$\mathcal{U}_{t}$, the 3D map during time slot~$t+1$, denoted by $\mathcal{D}^\text{e}_{t+1}$, evolves over two adjacent time slots as follows: 
    \begin{equation}\label{eq3}
        \mathcal{D}^\text{e}_{t+1} = \left\{ \mathcal{U}_{t} \cup \mathcal{D}^\text{e}_{t} \right\} \backslash \mathcal{C}_{t}, \,\, \forall t, t+1 \in \mathcal{T},
    \end{equation}
where~$\mathcal{C}_{t} \subseteq \mathcal{D}^\text{e}_{t}$ denotes the set of frames removed from set~$\mathcal{D}^\text{e}_{t}$ at the end of time slot~$t$. Determining set~$\mathcal{C}_{t}$ is a 3D map management decision in each time slot.

Additionally, the edge server needs to update the pose of the MAR device corresponding to each frame in set~$\mathcal{D}^\text{e}_{t}$ for real-time pose estimation. The edge server can update the pose in each frame by using maximum likelihood estimation based on co-visibility relationship among different frames~\cite{campos2021orb}, and both the updated set of frames stored in the 3D map and the estimated pose are sent back to the MAR device for updating the local 3D map, thereby supporting subsequent real-time pose estimation.

\subsection{Problem Formulation}

In this subsection, we first introduce the pose estimation uncertainty and then formulate a 3D map management problem to minimize the long-term pose estimation uncertainty. 

As mentioned in Subsection~\ref{sec23}, the edge server should estimate a set of poses corresponding to the frames stored in the 3D map based on their co-visibility relationship. The pose estimation uncertainty can quantify the impact of co-visibility relationship among the stored frames on the accuracy of estimated poses~\cite{chen2021cramer}.   

Let~$\hat{\bm{L}}(\mathcal{G}^\text{e}_{t})$ denote the reduced Laplacian matrix of a graph, i.e., 3D map~$\mathcal{G}^\text{e}_{t}$. The pose estimation uncertainty based on 3D map~$\mathcal{G}^\text{e}_{t}$ during time slot~$t$ can be approximated as follows~\cite{chen2023adaptslam,chen2021cramer}:
    \begin{equation}\label{eq5}
        u(\mathcal{G}^\text{e}_{t}) = - \log \left( \det ( \hat{\bm{L}}(\mathcal{G}^\text{e}_{t}) \otimes \boldsymbol{\Pi}) \right),
    \end{equation}
where notation~$\otimes$ represents the Kronecker product and~$\det(\cdot)$ denotes the determinant of a matrix. Matrix~$\boldsymbol{\Pi}$ in~\eqref{eq5} is constant and with the dimension of~$6\times6$ due to the six degrees of freedom of 3D pose, and the value of $\boldsymbol{\Pi}$ depends on the camera parameters of the MAR device. 

Then, we formulate a 3D map management problem to minimize the pose estimation uncertainty by selecting frames for uploading to update the 3D map. We define the obtained new 3D map if we add frame~$d$ into 3D map~$\mathcal{G}^\text{e}_{t} = (\mathcal{D}^\text{e}_{t}, \mathcal{E}^\text{e}_{t})$ as follows:
    \begin{equation}\label{}
        \mathcal{G}^\text{e}_{t} \cup \{ d \} := \left(\mathcal{D}^\text{e}_{t} \cup \{d\}, \mathcal{E}^\text{e}_{t}\cup \{e=(d,d')| d' \in \mathcal{D}^\text{e}_{t} \} \right),
    \end{equation}
where $\{e=(d,d')| d' \in \mathcal{D}^\text{e}_{t} \}$ is the set of newly generated edges due to adding frame~$d$. Denote the set of key frames, each containing a unique set of 3D map points~\cite{campos2021orb}, during time slot~$t+1$ by~$\mathcal{F}^\text{k}_{t+1} \subseteq \mathcal{F}_{t+1}$. Considering the dynamics of uplink quality and device pose, we formulate the long-term 3D map management problem as follows:
    \begin{subequations}\label{p1}
        \begin{align}
            \textrm{P1:} \,\, & \min_{ \{ \mathcal{U}_{t}, \mathcal{C}_{t} \}_{t \in \mathcal{T}} } \sum_{ t \in \mathcal T}{ |\mathcal{F}^\text{k}_{t+1}|^{-1} \sum_{f \in \mathcal{F}^\text{k}_{t+1}}{u(\mathcal{G}^\text{e}_{t} \cup \{ f \})} }\\
            \textrm{s.t.} &\,\, \eqref{eq1}, \eqref{eq2}, \eqref{eq3}, \\
            & \,\, \mathcal{U}_{t} \subseteq\mathcal{F}_{t}, \,\, \forall t \in \mathcal{T},\\     
            & \,\, \mathcal{C}_{t} \subseteq \mathcal{D}^\text{e}_{t} , \;\; \forall t \in \mathcal{T},
        \end{align}
    \end{subequations}
where $u(\mathcal{G}^\text{e}_{t} \cup \{ f \})$ represents the pose estimation uncertainty when a 3D map updated at time slot~$t$, i.e.,~$\mathcal{G}^\text{e}_{t}$, is used for estimating the pose of the MAR device at frame~$f$ during time slot~$t+1$. Problem~P1 is challenging due to two reasons. First, 3D map management for any given time slot relates to an NP-hard cardinality-fixed maximization problem~\cite{chen2023adaptslam}. Second, the 3D map management problem across multiple time slots is a sequential decision-making problem, in which the 3D map management decision for the current time slot affects those in successive time slots. Thus, making a 3D map management decision independently for each time slot may not be optimal in the long term.

\section{Proposed DT-based Solution}

To solve Problem~P1, we first transform it into a Markov decision process (MDP) problem and then propose a DT-based approach.  

\subsection{Problem Transformation}

We present the following lemma to show that the pose estimation uncertainty decreases when the number of frames stored in a 3D map increases.

\begin{lemma}\label{lemma1}
        If 3D map~$\mathcal{G}$ is connected, the value of~$u(\mathcal{G})$ monotonously decreases with the number of its nodes.
    \end{lemma}

    \begin{proof}

		Denote the 3D map when adding frame~$f'$ to 3D map~$\mathcal{G} = (\mathcal{D}, \mathcal{E})$ with a single extra edge, i.e., $|\mathcal{E}_{f'}| = 1$, and the 3D map when adding an extra edge~$e'$ between two existing frames in 3D map~$\mathcal{G}$ by~$\mathcal{G'} =  (\mathcal{D} \cup \{ f'\}, \mathcal{E} \cup \mathcal{E}_{f'})$ and $\mathcal{G''} =  (\mathcal{D}, \mathcal{E} \cup \{e'\})$, respectively. If~$\mathcal{G}$ is connected, both graphs~$\mathcal{G'}$ and~$\mathcal{G''}$ are connected. By comparing~$\hat{\bm{L}}(\mathcal{G})$ and the Laplacian expansions of~$\hat{\bm{L}}(\mathcal{G'})$ and~$\hat{\bm{L}}(\mathcal{G''})$, we can prove that~$\det(\hat{\bm{L}}(\mathcal{G})) < \det(\hat{\bm{L}}(\mathcal{G'}))$ and $\det(\hat{\bm{L}}(\mathcal{G})) < \det(\hat{\bm{L}}(\mathcal{G''}))$, respectively. Any 3D map~$\mathcal{G} \cup \{ f \}$ can be obtained by initially adding a frame to 3D map~$\mathcal{G}$ with a single extra edge, followed by iteratively repeating the operation of adding an extra edge between two existing frames. As a result, 3D map~$\mathcal{G} \cup \{ f \}$ is still connected, and~$\det(\hat{\bm{L}}(\mathcal{G})) < \det(\hat{\bm{L}}(\mathcal{G} \cup \{ f \}))$.

		The pose estimation uncertainty based on 3D map~$\mathcal{G}$ shown in~\eqref{eq5} can be given by:
			\begin{equation}\label{eqa}
		        \begin{aligned} 
		    		 u(\mathcal{G}) & = - \log \left( \det ( \hat{\bm{L}}(\mathcal{G}) \otimes \boldsymbol{\Pi} ) \right)\\
		            & = - \log \left( \det ( \hat{\bm{L}}(\mathcal{G}))^{|\mathcal{D}|-1} \det(\boldsymbol{\Pi})^{6} \right)
		        \end{aligned} 
			\end{equation}
		where $|\mathcal{D}|$ denotes the number of frames in 3D map~$\mathcal{G}$, and the dimension of~$\tilde{\bm{L}}(\mathcal{G})$ is $(|\mathcal{D}|-1) \times (|\mathcal{D}|-1)$. According to~\eqref{eqa}, we can prove that $u(\mathcal{G}) < u(\mathcal{G} \cup \{ f \})$. 
    \end{proof}

Since the pose estimation uncertainty based on an unconnected 3D map can approach positive infinity, we must keep the 3D map connected. Lemma~\ref{lemma1} allows us to reduce the solution space of Problem~P1 since the optimal solution to Problem~P1 must keep storing~$D$ frames at the edge server for any given time slot, i.e.,
   \begin{equation}\label{eq7}
		|\left\{ \mathcal{U}_{t} \cup \mathcal{D}^\text{e}_{t} \right\} \backslash \mathcal{C}_{t} | = D, \,\, \forall t \in \mathcal{T}.		 
   \end{equation}
    
To solve Problem~P1, we model the sequential decision making on 3D map management as an MDP. Denote the state space and the action space by $\mathcal{O}$ and $\mathcal{A}$, respectively. Let $o_{t} = [\mathcal{G}^\text{e}_{t-T}, \mathcal{F}_{t-T}, R_{t-T}, \cdots, \mathcal{G}^\text{e}_{t}, \mathcal{F}_{t}, R_{t}] \in \mathcal{O}$ and $a_{t} = [\mathcal{U}_{t}, \mathcal{C}_{t}] \in \mathcal{A}$ denote the state and the action, i.e., 3D map management decision, at the beginning of time slot~$t$, respectively. We define the reward function of action~$a_{t}$ on state~$o_{t}$ as the negative value of pose estimation uncertainty based on the updated 3D map, as follows:
   \begin{equation}\label{eq10}
		r(o_{t}, a_{t}) =  - |\mathcal{F}^\text{k}_{t+1}|^{-1} \sum_{f \in \mathcal{F}^\text{k}_{t+1}}{u(\mathcal{D}^\text{e}_{t} \cup \{ f \})}.		 
   \end{equation}
Problem~P1 can be reformulated as the following discounted MDP problem~\cite{zhou2020deep}:
    \begin{subequations}\label{p3}
        \begin{align}
            \textrm{P2:} \,\, & \max_{ \{ a_{t} \}_{t \in \mathcal{T}} } \sum_{ t \in \mathcal T}{ \gamma^{t} r(o_{t}, a_{t}) }\\
            \textrm{s.t.} &\,\, \eqref{eq1}, \eqref{eq3}, (\ref{p1}\text{c}), (\ref{p1}\text{d}), \eqref{eq7},
        \end{align}
    \end{subequations}
where~$\gamma \in (0,1)$ is the discount factor for quantifying the impact of an action on the rewards obtained in future time slots. Our goal is to find a policy, i.e.,~$\pi(o)$, for making proper 3D map management decisions in each state.

\subsection{DT-based 3D Map Management}\label{sec32}

Generally, reinforcement learning (RL) is an effective method for addressing sequential decision making problems in dynamic environments. However, conventional model-free RL algorithms face the low data-efficiency issue in finding the optimal policy when the state space is high-dimensional~\cite{zhou2020deep,li2020deep}. In Problem~P2, the lower bound of the dimension of each state is~$D^{\binom{|\mathcal{M}_{f}|}{|\mathcal{M}|}}$ where~$\binom{|\mathcal{M}_{f}|}{|\mathcal{M}|}$ represents the number of combinations of detecting~$|\mathcal{M}_{f}|$ 3D map points in a frame from the set of all 3D map points~$\mathcal{M}$. To address this issue, we propose a DT-based approach by using MBRL, consisting of the following four modules. 

\subsubsection{Real experience collection} A DT created for the MAR device is deployed at the edge server. The DT contains the real experiences of 3D map management collected in past time slots for learning the 3D map management policy in subsequent time slots. A real experience collected in time slot~$t+1$ is defined as a tuple~$\xi_{t} = (o_{t}, a_{t}, r(o_{t}, a_{t}), o_{t+1})$. Let~$\Xi^\text{r}$ represent the set of collected real experiences contained in the DT, which can be updated by the DT per time slot by collecting a new real experience.

\subsubsection{Map point prediction} The collected real experiences in~$\Xi^\text{r}$ are used to predict state transition, in addition to learning the 3D map management policy as model-free RL. According to~\eqref{eq3}, the basis of predicting state transition in this problem is to predict the associated 3D map points of all frames during the subsequent time slot, i.e.,~$\{\mathcal{M}_{f} \}_{f \in \mathcal{F}_{t}}$. We use a recurrent neural network (RNN) to approximate the mapping from the set of 3D map points detected in the frames during time slots~$t-1$ to $t-T$ to the set of 3D map points detected in the frames during time slot~$t$, as follows:
   \begin{equation}\label{eq12}
		 \{\tilde{\mathcal{M}}_{f} \}_{f \in \tilde{\mathcal{F}}_{t}} = \phi \left(\left\{\mathcal{M}_{f'} \right\}_{f' \in \mathcal{F}_{t-1} \cup \cdots \cup \mathcal{F}_{t-T} }; \boldsymbol{\theta} \right), 	 
   \end{equation}
where~$\{\tilde{\mathcal{M}}_{f} \}_{f \in \tilde{\mathcal{F}}_{t}}$ is the predicted 3D map points associated with the frames during time slot~$t$, and $\boldsymbol{\theta}$ denotes the parameters of the RNN. Based on the collected real experiences from~$\Xi^\text{r}$, the optimal parameters of the RNN are trained by using the following equation:
   \begin{equation}\label{eq13}
		 \boldsymbol{\theta}^{*} = \arg \min_{ \left\{\boldsymbol{\theta} \right\} } \| \left\{\mathcal{M}_{f} \right\}_{f \in \mathcal{F}_{t}} -  \{\mathcal{\tilde{M}}_{f} \}_{f \in \tilde{\mathcal{F}}_{t}}  \|^{2}. 	 
   \end{equation}
Parameter~$\boldsymbol{\theta}$ can be updated when real experiences contained in the DT are updated.

\subsubsection{Artificial experience generation} The RNN enables the DT to emulate the state transition given a 3D map management decision, thereby generating artificial experiences. Specifically, given a state~$\tilde{o} \in \mathcal{O}$, an artificial action~$\tilde{a} \in \mathcal{A}$ can be determined based on a random policy~$\tilde{\pi}(\tilde{o})$. Assuming how the set of representative frames~$\tilde{\mathcal{F}}^\text{k}_{t+1}$ is selected from~$\tilde{\mathcal{F}}_{t+1}$ known~\emph{a priori}, we can calculate the reward given~$\tilde{o}$ and~$\tilde{a}$, i.e.,~$r(\tilde{o}, \tilde{a})$, by~\eqref{eq10}. For emulating the next state~$\tilde{o}^{'} \in \mathcal{O}$, the DT predicts the set of frames and their associated 3D map points during subsequent time slots by using the RNN as shown in~\eqref{eq12}, thereby obtaining the updated 3D map based on~\eqref{eq3}. The uplink data rate in the next state can be predicted based on the modeled two-state Markov chain presented in Subsection~\ref{sec22}. Since action~$\tilde{a}$ is not actually implemented for 3D map management, we define such an artificial experience as a tuple~$\tilde{\xi} = (\tilde{o}, \tilde{a}, r(\tilde{o}, \tilde{a}), \tilde{o}^{'} ) \in \Xi^\text{a}$ where~$\Xi^\text{a}$ is a set of artificial experiences also contained in the DT.

	\begin{algorithm}[t] 
		\caption{AMM Algorithm}\label{alg1}
		\LinesNumbered
		\textbf{Input:} $N$ and $|\mathcal{I}|$\\
		\textbf{Initialization:} $\boldsymbol{\vartheta}^\text{c}$, $\boldsymbol{\vartheta}^\text{a}$, $\boldsymbol{\theta}$, $o_{1}$\\
		\For{$t \in \mathcal{T}$}
		{	
			Select action $a_{t} = \pi \left( o_{t}; \boldsymbol{\vartheta}^\text{a} \right)$;\\
			$r(o_{t}, a_{t})$, $o_{t+1}$ $\leftarrow$ take action $a_{t} $ on state~$o_{t}$;\\
			$\Xi^\text{r}$ $\leftarrow$ update real experiences with tuple $(o_{t}, a_{t}, r(o_{t}, a_{t}), o_{t+1})$;\\

			Randomly sample $|\mathcal{I}|$ tuples from $\Xi^\text{r}$;\\
			Set $y_{i} = r(o_{i}, a_{i}) + \gamma Q (o_{i}, \pi(o_{i}; \boldsymbol{\vartheta}^\text{a}); \boldsymbol{\vartheta}^\text{c})$; \\
			Update $\boldsymbol{\vartheta}^\text{c}$ by minimizing $\frac{1}{|\mathcal{I}|} \sum_{\mathcal{I}} \left(y_{i} - Q(o_{i}, a_{i}; \boldsymbol{\vartheta}^\text{c}) \right)^{2}$;\\
			Update $\boldsymbol{\vartheta}^\text{a}$ by using policy gradient descent in~\eqref{eq15};\\

			Update $\boldsymbol{\theta}$ by optimizing~\eqref{eq13} based on $\Xi^\text{r}$;\\

			$\Xi^\text{a}$ $\leftarrow$ Update artificial experiences based on $\Xi^\text{r}$;\\
			
			\For{$n = 1, \cdots N$}
			{
				Randomly sample $|\mathcal{I}|$ tuples from $\Xi^\text{a}$;\\
				Update $\boldsymbol{\vartheta}^\text{c}$ and $\boldsymbol{\vartheta}^\text{a}$ by repeating lines~9-10 based on the sampled artificial experiences;\\

			} 
			$t$, $o_{t}$ $\leftarrow$ $t+1$, $o_{t+1}$;\\   	
		}
		\textbf{Output:} $\pi^{*}(o;\boldsymbol{\vartheta}^\text{a})$

	\end{algorithm}

\subsubsection{MBRL based on blended experiences} Leveraging the combination of the real and the artificial experiences, referred to as \emph{blended experiences}, we propose the adaptive map management (AMM) algorithm based on MBRL to learn a long-term 3D map management policy~i.e.,~$\pi(o)$.

Define a Q-value function of state~$o_{t}$ and action~$a_{t}$ as the accumulated discounted reward, as follows: 
	\begin{equation}\label{}
		Q(o = o_{t}, a = a_{t}) =  \sum_{k = 1}^{K}{\gamma^{k} r(o_{t+k+1}, a_{t+k+1})},
	\end{equation}
where the Q-value quantifies the long-term impact of each action on the subsequent states as well as actions~\cite{zhou2020deep}. Due to the high-dimensional state space, we use a DNN with parameter~$\boldsymbol{\vartheta}^\text{c}$ to approximate the Q-value function, i.e.,~$Q(o, a; \boldsymbol{\vartheta}^\text{c})$.
Moreover, we approximate the 3D map management policy~$\pi(o; \boldsymbol{\vartheta}^\text{a})$ by using another DNN with parameter~$\boldsymbol{\vartheta}^\text{a}$. Given~$Q(o, a; \boldsymbol{\vartheta}^\text{c})$, the policy gradient is given for finding the parameter of the optimal policy~$\pi^{*}$:
	\begin{equation}\label{eq15}
		\nabla_{\boldsymbol{\vartheta}^\text{a}} \Omega (\boldsymbol{\vartheta}^\text{a}) = \frac{1}{|\mathcal{I}|} \sum_{i \in \mathcal{I}} \nabla_{\mathbf{w}} Q \left(o_{i}, a_{i} ; \boldsymbol{\vartheta}^\text{c} \right) \big|_{\pi(o_{i}; \boldsymbol{\vartheta}^\text{a})} \nabla_{ \boldsymbol{\vartheta}^\text{a} }\pi\left( o_{i}; \boldsymbol{\vartheta}^\text{a} \right),
	\end{equation}
where $\mathcal{I}$ denotes a set of tuples that can be selected from~$\Xi^\text{r}$,~$\Xi^\text{a}$, or both, and $\Omega (\boldsymbol{\vartheta}^\text{a})$ represents the objective function in Problem~P2 in terms of parameters~$\boldsymbol{\vartheta}^\text{a}$. To learn the impact of co-visibility relationship among frames on 3D map management, we leverage a graph convolutional network (GCN) for building the two DNNs. 

The developed AMM algorithm featuring MBRL is presented in Algorithm~\ref{alg1}. In lines~1-2, we set the batch size~$|\mathcal{I}|$ and determine the value of~$N$ for adjusting the ratio of real and artificial experiences. Procedures in lines~3-10 correspond to conventional model-free RL methods based on real experiences~\cite{zhou2020deep}. In lines~11-12, the DT trains the RNN by using the collected real experiences and generate artificial experiences. The parameters of two DNNs for approximating the Q-value and policy functions are optimized by using the generated artificial experiences in lines~13-16. Due to the consideration of long-term impact of 3D map management, the developed MBRL-based algorithm can adapt to time-varying pose of the MAR device and dynamic uplink data rate. Moreover, the developed AMM algorithm is more data-efficient than model-free RL since the real experiences can be utilized to generate lots of artificial experiences, thereby accelerating the exploration process in a high-dimensional state space.

\section{Performance Evaluation}

\subsection{Simulation Settings}

\begin{table}[t]
	\footnotesize 
	\centering
	\captionsetup{justification=centering,singlelinecheck=false}
	\caption{Simulation Parameters}\label{table1}
	\begin{tabular}{c|c|c|c}
		\hline\hline
		 Parameter & Value & Parameter & Value\\
		 \hline\hline
		 $\tau$ & 2 sec. & $d$ & 2\,Mbits \\
		 \hline
		 $R_\text{h}$& 20\,Mbits/sec. & $R_\text{l}$ & 8\,Mbits/sec.\\
		 \hline
		 $D$ & 25 frames& $|\mathcal{F}_{t}|$ & 60 frames\\
		 \hline
	\end{tabular}
	\vspace{-0.5cm}
\end{table}

    \begin{figure}[t]
        \centering
        \includegraphics[width=0.35\textwidth]{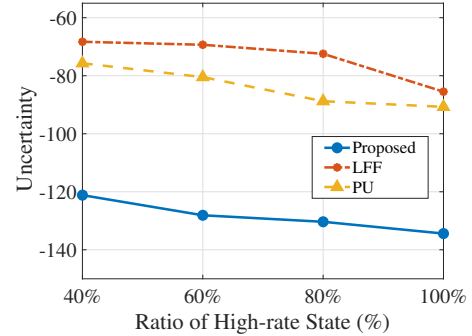}
        \caption{Performance comparison in different uplink data rates.}\label{fig1}
        \vspace{-0.6cm}
    \end{figure}

We use data collected in a real indoor environment, i.e., the ``westgate-playroom'' camera frame sequence in the SUN3D dataset~\cite{xiao2013sun3d}. The set of 3D map points detected in each frame~$\mathcal{M}_{f}$ and the set of representative frames in each time slot~$\mathcal{F}^\text{k}_{t}$ are obtained using the open-source ORB-SLAM framework~\cite{campos2021orb}. Other parameters are listed in Table~I.

For the DNNs in the developed AMM algorithm, a GCN, i.e., two graph convolutional layers with (128, 32) neurons, is utilized as an embedding layer. Following the embedding layer, three fully-connected layers with (64, 32, 32) neurons and four fully-connected layers with (64, 32, 16, 4) neurons are used to approximate the Q-value function and the policy function, respectively. For the RNN used in artificial experience generation, we utilize three long-short-term-memory (LSTM) layers with (128, 32, 16) neurons, followed by two full-connected layers with (512, 64) neurons.

We adopt the following two popular 3D map management schemes in MAR as benchmark~\cite{campos2021orb,apicharttrisorn2020characterization}:
    \begin{itemize}
        \item \emph{Latest Frame First (LFF):} The 3D map is updated by using the lastly captured frames;

        \item \emph{Periodical uploading (PU):} Frames are selected evenly from a frame sequence within a time slot, and the earliest captured frames are removed from the 3D map. 
    \end{itemize}

\subsection{Performance of DT-based 3D Map Management}

We divide the camera frame sequence \bl{into} two datasets: a training dataset containing 11,000 frames and a testing dataset comprising 2,500 frames. First, we evaluate the performance of the designed DT-based scheme with well-trained DNNs on the training dataset, as shown in Fig.~\ref{fig1}. Subsequently, we run the scheme with the DNNs on the testing dataset to evaluate the performance of the DT-based scheme in an unknown environment, as shown in Fig.~\ref{fig2}.

\bl{In Fig.~\ref{fig1}, we compare the performance of our DT-based scheme with that of two benchmark schemes under various network dynamics.} Each point in the figure is averaged \bl{over} 15 independent simulation runs. By setting different transition matrices for the two-state Markov chain determining the uplink data rate, we change the ratio of the high-rate state over all time slots. We observe that the proposed scheme can maintain a 3D map with a lower pose estimation uncertainty compared to the benchmark schemes under various ratios of high-rate state. This is because the DT-based scheme utilizing MBRL can learn a 3D map management policy by taking into account the long-term impact of each action in a dynamic environment, as opposed to the myopic 3D map management adopted by the benchmark schemes.

    \begin{figure}[t]
        \centering
        \includegraphics[width=0.33\textwidth]{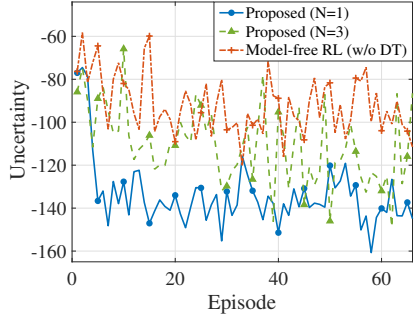}
        \caption{Convergence in an unknown environment.}\label{fig2}
        \vspace{-0.6cm}
    \end{figure}

\bl{In Fig.~\ref{fig2}, we compare the convergence performance of our DT-based scheme with that of a model-free RL, i.e., DDPG, without using the DT (labeled as "Model-free RL (w/o DT)") on the testing dataset.} Each episode spans 2,500 frames. We can make two observations from this figure. First, we observe that the designed scheme converges to a better policy in fewer episodes than benchmark schemes. This is because leveraging additional emulated data, i.e., artificial experiences provided by the DT, in training can accelerate policy exploration in an unknown environment. Second, we adjust the amount of emulated data used in each training iteration as mentioned in Subsection~\ref{sec32}, corresponding to the parameter~$N$ in the legend of Fig~\ref{fig2}. We observe that the policy learned by the AMM algorithm may not be stationary when an increasing amount of emulated data is used due to the limited accuracy of the RNN used for generating artificial experiences. Therefore, the amount of emulated data used in training should be carefully adjusted.

\section{Conclusion and Future Work}

In this paper, we have designed a DT-based 3D map management scheme to facilitate pose estimation in edge-assisted MAR. The DT created for the MAR device can emulate 3D map management based on predicting subsequent frames and provide extensive emulated data. For minimizing pose estimation uncertainty, our MBRL algorithm learned a 3D map management policy based on the data collected from both the actual and the emulated 3D map management. The results have demonstrated the adaptivity and data efficiency of the DT-based scheme in dynamic environments. The designed scheme establishes a foundation for customizing DTs to differentiate 3D map management policies for different MAR devices. In the future, we will extend 3D map management to encompass multiple MAR devices while considering their unique time-varying poses and resource constraints.

\bibliography{ref}

\bibliographystyle{IEEEtran}

\end{document}